\documentclass[runningheads]{llncs}
\usepackage[utf8]{inputenc}
\usepackage[T1]{fontenc}
\usepackage{amsmath,amssymb}
\usepackage{microtype}
\usepackage{booktabs}
\usepackage{tabularx}
\usepackage{placeins}
\makeatletter
\@ifclassloaded{llncs}{}{%
  \usepackage{amsthm}
  \newtheorem{theorem}{Theorem}\newtheorem{lemma}[theorem]{Lemma}
  \newtheorem{proposition}[theorem]{Proposition}\newtheorem{corollary}[theorem]{Corollary}
  \theoremstyle{definition}\newtheorem{definition}[theorem]{Definition}
  \theoremstyle{remark}\newtheorem{remark}[theorem]{Remark}
  \newcommand{\institute}[1]{\gdef\@inst{#1}}\newcommand{\email}[1]{\texttt{#1}}
  \newcommand{\titlerunning}[1]{}\newcommand{\authorrunning}[1]{}
  \newcommand{\keywords}[1]{\par\medskip\noindent\textit{Keywords:} {\def\and{\unskip\ $\cdot$\ }#1}}
}
\makeatother

\newcommand{\F}{\mathbb{F}}
\newcommand{\rk}{\operatorname{rk}}
\newcommand{\Tr}{\operatorname{Tr}}
\newcommand{\Ker}{\operatorname{Ker}}
\newcommand{\Gab}{\mathrm{Gab}}
\newcommand{\GabDec}{\mathsf{GabDec}_1}
\newcommand{\HomMQ}{\mathsf{HomMQ}}
\newcommand{\Piv}{\mathsf{TSM}}

\begin{document}

\title{Maximum-Likelihood Decoding of Gabidulin Codes Is NP-Hard Even in Dimension One}
\titlerunning{ML decoding of Gabidulin codes is NP-hard}
\author{Philippe Gaborit}
\authorrunning{P. Gaborit}
\institute{XLIM, Universit\'e de Limoges, France\\ \email{gaborit@unilim.fr}}
\maketitle

\begin{abstract}
Fix a prime power $q$. We prove that deciding whether a received word $y\in\F_{q^m}^n$ lies at rank distance at most $n-2$ from the one-dimensional Gabidulin code $\{\lambda g:\lambda\in\F_{q^m}\}$ is NP-complete. Thus maximum-likelihood decoding of Gabidulin codes is NP-hard even for $k=1$, an analogue of the Guruswami--Vardy theorem for Reed--Solomon codes. In the Hamming metric, by contrast, hardness at this radius requires the dimension to grow with the length, as in the Guruswami--Vardy theorem. The reduction is deterministic; consequently Rank Syndrome Decoding ($\mathsf{RSD}$) and Rank Minimum Distance ($\mathsf{RMD}$) for $\F_{q^m}$-linear codes are NP-complete under deterministic polynomial reductions, whereas the previously known reduction was randomized. The proof goes through a pivot problem that we introduce, Twisted Square Membership ($\Piv$): given $\F_q$-subspaces $U,W\subset\F_{q^N}$, decide whether some non-zero $a\in U$ satisfies $a^{q+1}\in W$. Quadratic systems in $d$ unknowns are encoded into $\Piv$, then $\Piv$ into Gabidulin decoding. The last step views the received word as an $\F_q$-linear map and decoding at radius $n-2$ as the search for an $\F_q$-plane on which it acts as a scalar. With support $U\oplus\tau W$ in a cubic extension and map $a+\tau b\mapsto a^{q^2}+\tau b^q$, two Frobenius twists force solution planes to encode the relation $a^{q+1}\in W$. All constructions are explicit and polynomial, with $N=O(qd^2)$, $n\le N+d$, and $m=3N$.
\end{abstract}

\section{Introduction}\label{sec:intro}

\subsection{Context and motivation}

Gabidulin codes \cite{Del78,Gab85} are the rank-metric analogues of Reed--Solomon codes. They are maximum rank distance codes and can be decoded in polynomial time up to $\lfloor(n-k)/2\rfloor$ errors. Beyond half the minimum distance, however, the picture changes sharply: no analogue of the Guruswami--Sudan algorithm is known, and for many parameters no efficient worst-case list decoder can exist because some received words have exponentially many nearby codewords \cite{WZ13,RWZ16}. The complexity of maximum-likelihood decoding for Gabidulin codes themselves has remained open, to the best of our knowledge. Our main result closes this gap in a particularly rigid case: for every fixed prime power $q$, decoding is NP-complete already for dimension-one Gabidulin codes at radius $n-2$ (Theorem~\ref{thm:main}).

This result fits into a longer line of hardness results for decoding. In the Hamming metric, Berlekamp, McEliece and van Tilborg proved that decoding a general linear code is NP-complete \cite{BMvT78}; Vardy settled the corresponding minimum-distance problem \cite{Var97}, and Dumer, Micciancio and Sudan proved hardness of approximation \cite{DMS03}. Structured algebraic codes require separate arguments. Reed--Solomon codes admit polynomial-time decoding up to half the minimum distance, and the Guruswami--Sudan algorithm reaches the Johnson radius \cite{GS99}; nevertheless, Guruswami and Vardy proved maximum-likelihood decoding NP-hard at radius $n-k-1$ over fields of exponential size \cite{GV05}. Further hardness results connect bounded-distance decoding to discrete logarithms \cite{CW07}, cover algebraic-geometry codes under randomized reductions \cite{Che08}, and lower the NP-hard radius for Reed--Solomon codes to $n-k-c\log n/\log\log n$ \cite{GGG15}.

In the rank metric, a vector in $\F_{q^m}^n$ is weighted by the $\F_q$-dimension of the space spanned by its coordinates rather than by the number of non-zero coordinates. Besides their coding-theoretic role, Gabidulin codes underlie error control in random linear network coding \cite{KK08,SKK08} and have found applications in theoretical computer science, from explicit hitting sets for low-rank matrices \cite{FS12} to list decoding via subspace designs \cite{GX13,GK13}; random Gabidulin codes achieve list decoding capacity \cite{GXYZ24}. On the complexity side, MinRank is NP-complete \cite{BFS99}, and decoding an $\F_q$-linear code is an instance of MinRank. The $\F_{q^m}$-linear case is substantially more rigid. Gaborit and Z\'emor proved hardness for such rank-metric decoding problems through a randomized reduction from Hamming decoding \cite{GZ16}, but no hardness result was known for maximum-likelihood decoding of Gabidulin codes themselves, to the best of our knowledge.

The question also matters in rank-based cryptography. The first rank-metric encryption scheme, due to Gabidulin, Paramonov and Tretjakov \cite{GPT91}, is a McEliece-type construction based on a disguised Gabidulin code; Overbeck's structural attacks \cite{Ove08} motivated repaired variants, including Loidreau's scheme \cite{Loi17} and its subsequent cryptanalysis \cite{CC20}. More directly related to this work, the Faure--Loidreau scheme \cite{FL05} and its repaired version LIGA \cite{RPW21} rely on decoding Gabidulin codes beyond the unique-decoding radius. Our result concerns worst-case complexity only and does not address the particular parameters or instance distributions used by these cryptosystems; it places the underlying worst-case decoding problem on a firmer complexity-theoretic footing. Other rank-based constructions rely on rank syndrome decoding for random or quasi-cyclic $\F_{q^m}$-linear codes, including LRPC-based schemes \cite{GMRZ13,ABD+19}, the ROLLO and RQC candidates to the NIST post-quantum standardization process \cite{ROLLO,RQC}, and rank-based signatures such as Durandal \cite{ABG+19}.

\subsection{The idea of the method}

The proof rests on three elementary observations whose combination is the key.

\emph{First, decoding becomes an eigenplane problem.} For a dimension-one Gabidulin code $\{\lambda g\}$, the received word $y$ can be viewed as an $\F_q$-linear map $\tilde y$ on the $\F_q$-space spanned by the coordinates of $g$, while the codeword $\lambda g$ corresponds to multiplication by $\lambda$. Decoding at radius $n-2$ is therefore equivalent to finding an $\F_q$-plane on which $\tilde y$ acts as a scalar, which we call an \emph{eigenplane}.

\emph{Second, the eigenplane equation produces a twisted square.} Because the received word is part of the constructed instance, we let $\tilde y$ act by different Frobenius twists on different parts of the support. Comparing $a\mapsto a^{q^2}$ and $b\mapsto b^q$ along an eigenplane yields $b\in\F_q^*\,a^{q+1}$. Thus the only non-linearity that must be encoded is the map $a\mapsto a^{q+1}=a\cdot a^q$. This leads to the pivot problem Twisted Square Membership ($\Piv$): given subspaces $U,W\subset\F_{q^N}$, decide whether there exists a non-zero $a\in U$ with $a^{q+1}\in W$. The linear membership constraints are built into the support, while the eigenplane equation supplies the non-linear relation.

\emph{Third, twisted squares encode homogeneous quadratic systems.} The element $a\cdot a^q$ is quadratic in the $\F_q$-coordinates of $a$. With a suitable basis of $U$, every homogeneous quadratic system can be represented by linear conditions on $a^{q+1}$. Since such systems encode 3-SAT, $\Piv$ is NP-complete.

The complete reduction is therefore
\[
\textsf{3-SAT}\longrightarrow\HomMQ\longrightarrow\Piv\longrightarrow\GabDec,
\]
with $\Piv$ separating the finite-field encoding of quadratic forms from the Gabidulin decoding construction.

\subsection{Contributions and organization}

Theorem~\ref{thm:main} proves that, for every fixed prime power $q$, deciding whether a received word lies at rank distance at most $n-2$ from a dimension-one Gabidulin code is NP-complete. Radius $n-1$ is always reachable, so $n-2$ is the first nontrivial radius. Maximum-likelihood decoding is therefore NP-hard as soon as the problem ceases to be trivial, already for $k=1$. This is a rank-metric analogue of the Guruswami--Vardy theorem, with the notable difference that its Hamming-metric counterpart requires growing dimension; see Section~\ref{sec:cons}.

The reduction is deterministic. As a consequence, rank syndrome decoding and rank minimum distance for general $\F_{q^m}$-linear codes are NP-complete under deterministic polynomial reductions, with rank minimum distance already hard for dimension-two codes (Corollaries~\ref{cor:sd} and~\ref{cor:md}). The previously known reduction for these problems was randomized \cite{GZ16}.

A further contribution is the pivot problem $\Piv$, a compact NP-complete problem over finite fields that may be of independent interest. The construction is fully explicit: the encoding exponents are given by a closed formula obtained by interleaving base-$q$ expansions, and a quadratic system in $d$ unknowns leads to parameters $N=O(qd^2)$, $n\le N+d$, and $m=3N$.

\emph{Organization.} Section~\ref{sec:prelim} recalls the necessary background and develops the eigenplane formulation and $\Piv$. Section~\ref{sec:overview} states the main theorem and gives the complete reduction at a high level. Sections~\ref{sec:sat}--\ref{sec:gab} prove its three steps. Section~\ref{sec:cons} derives the complexity consequences, and Section~\ref{sec:disc} discusses open questions and suggests some directions.

\section{Background and preliminaries}\label{sec:prelim}

This section gathers what the reduction needs. The first part recalls standard material: the rank metric, Gabidulin codes, the two decoding problems of rank-metric coding theory, and the two classical NP-complete problems, 3-SAT and homogeneous quadratic systems, from which we start. The second part contains the two ingredients specific to this paper: the reformulation of decoding in dimension one as the search for an eigenplane, and the pivot problem $\Piv$.

\subsection{Background: known problems}

Throughout, $q$ is a fixed prime power. For $v=(v_1,\dots,v_n)\in\F_{q^m}^n$, the rank $\rk(v)$ is the dimension over $\F_q$ of $\langle v_1,\dots,v_n\rangle_{\F_q}$; equivalently it is the rank of the $m\times n$ matrix over $\F_q$ whose $j$-th column is the vector of coordinates of $v_j$ in a fixed basis of $\F_{q^m}/\F_q$. The rank distance is $d(u,v)=\rk(u-v)$. For instance, with $\alpha\in\F_{q^m}\setminus\F_q$, the vector $(\alpha,\alpha,\alpha+1)$ has Hamming weight $3$ but rank $2$, its coordinates spanning $\langle 1,\alpha\rangle$: vanishing coordinates in the Hamming metric are replaced by $\F_q$-linear dependencies in the rank metric.

An $\F_{q^m}$-linear code of length $n\le m$ and dimension $k$ has minimum rank distance at most $n-k+1$; this is the Singleton bound for the rank metric \cite{Del78,Gab85}, and codes meeting it are called maximum rank distance (MRD) codes. Let $g\in\F_{q^m}^n$ with $\rk(g)=n$. The Gabidulin code of dimension $k$ is $\Gab(g,k)=\{(f(g_1),\dots,f(g_n))\}$ where $f$ runs over the $q$-polynomials $\sum_{i<k}f_iX^{q^i}$ with $f_i\in\F_{q^m}$; it is MRD, with minimum rank distance $n-k+1$. For $k=1$ one has $f(X)=\lambda X$, so $\Gab(g,1)=\{\lambda g:\lambda\in\F_{q^m}\}$ has minimum distance $n$ and unique decoding radius $\lfloor (n-1)/2\rfloor$. At the other end, decoding $\Gab(g,1)$ at radius $n-1$ is easy: every received word lies at rank distance at most $n-1$ from the code, since for any coordinate $y_j$ the scalar $\lambda=y_j/g_j$ makes the $j$-th coordinate of $y-\lambda g$ vanish, so that its coordinates span a space of dimension at most $n-1$. The first radius at which decoding is a non-trivial question is therefore $n-2$.

We shall decode this code, but our conclusions concern the two general problems of rank-metric coding theory, which we recall now. An $\F_{q^m}$-linear code of length $n$ and dimension $k$ is given by a generator matrix $G\in\F_{q^m}^{k\times n}$ or by a parity-check matrix $H\in\F_{q^m}^{(n-k)\times n}$. Both problems below are in NP, with $e$ or the codeword as certificate.

\begin{definition}[rank syndrome decoding, $\mathsf{RSD}$]
Given $H\in\F_{q^m}^{(n-k)\times n}$, a syndrome $s\in\F_{q^m}^{n-k}$ and an integer $t$, decide whether there exists $e\in\F_{q^m}^n$ with $He^\top=s$ and $\rk(e)\le t$.
\end{definition}

\begin{definition}[rank minimum distance, $\mathsf{RMD}$]
Given $G\in\F_{q^m}^{k\times n}$ and an integer $w$, decide whether the code generated by $G$ contains a non-zero codeword of rank at most $w$.
\end{definition}

$\mathsf{RSD}$ is equivalent to bounded-distance decoding, with input $(G,y,t)$ and question whether some codeword lies at rank distance at most $t$ from $y$, since $e=y-c$ and $H$, $s=Hy^\top$ are computed from $G$ and $y$ in polynomial time.

\begin{remark}[relation with MinRank]
Identifying a vector of $\F_{q^m}^n$ with its $m\times n$ matrix of coordinates over $\F_q$, the codewords of an $\F_{q^m}$-linear code form an $\F_q$-space of matrices, and $\mathsf{RSD}$ asks for a matrix of rank at most $t$ in an affine space of matrices: this is an instance of MinRank. It is a structured one, since the matrix space is stable under multiplication by $\F_{q^m}$; for the code $\{\lambda g\}$, for instance, every non-zero matrix of the space has full rank $n$. This gives $\mathsf{RSD}\le$ MinRank and not the converse, so the hardness of MinRank \cite{BFS99} does not transfer and a dedicated reduction is needed.
\end{remark}

The source of our reduction is the satisfiability problem for formulas in 3-conjunctive normal form.

\begin{definition}[3-SAT]
An instance consists of Boolean variables $b_1,\dots,b_v$ and a formula $\Phi=C_1\wedge\dots\wedge C_c$. Each clause $C_j$ is a disjunction $\ell_{j,1}\vee\ell_{j,2}\vee\ell_{j,3}$ of three \emph{literals}, a literal being either a variable $b_i$ or its negation $\neg b_i$. An \emph{assignment} is a map $\sigma:\{b_1,\dots,b_v\}\to\{\text{true},\text{false}\}$; it makes the literal $b_i$ true if $\sigma(b_i)=\text{true}$, the literal $\neg b_i$ true if $\sigma(b_i)=\text{false}$, a clause true if at least one of its three literals is true, and $\Phi$ true if all clauses are true. The question is whether there exists an assignment making $\Phi$ true, in which case $\Phi$ is \emph{satisfiable}.
\end{definition}

For instance $\Phi=(b_1\vee\neg b_2\vee b_3)\wedge(\neg b_1\vee b_2\vee b_3)$ is satisfiable, by $\sigma(b_3)=\text{true}$, whatever the values of $b_1$ and $b_2$. 3-SAT is the NP-complete problem of Cook and Levin \cite{GJ79}; the size of an instance is polynomial in $v+c$. Its first translation is into quadratic systems, whose NP-completeness is classical as well \cite{FY79,GJ79}; we use the homogeneous version.

\begin{definition}[$\HomMQ$]
Given homogeneous quadratic forms $Q_1,\dots,Q_s$ in $d$ unknowns over $\F_q$, decide whether there exists $x=(x_1,\dots,x_d)\in\F_q^d$, $x\neq0$, with $Q_1(x)=\dots=Q_s(x)=0$.
\end{definition}

\subsection{Preliminaries: the eigenplane formulation and the TSM pivot problem}

We next specialize to the problem studied in this paper: decoding dimension-one Gabidulin codes at the first nontrivial radius.

\begin{definition}[$\GabDec$]
Given $g,y\in\F_{q^m}^n$ with $\rk(g)=n$, decide whether there exists $\lambda\in\F_{q^m}$ such that $\rk(y-\lambda g)\le n-2$.
\end{definition}

This is $\mathsf{RSD}$ for $k=1$, $G=g$ and $t=n-2$; it is in NP with $\lambda$ as certificate. Radius $n-1$ is always reachable, as shown below, so $n-2$ is the first non-trivial radius.

The key to everything that follows is to read vectors as linear maps. Throughout the paper, $g=(g_1,\dots,g_n)\in\F_{q^m}^n$ has $\F_q$-linearly independent coordinates, and we call
\[
G=\langle g_1,\dots,g_n\rangle_{\F_q}\subset\F_{q^m}
\]
the \emph{support} of the code $\Gab(g,1)$; it is an $\F_q$-space of dimension $n$, of which $g$ is a basis.

\begin{definition}[the map $\tilde y$]\label{def:ytilde}
For $y=(y_1,\dots,y_n)\in\F_{q^m}^n$, we denote by $\tilde y:G\to\F_{q^m}$ the unique $\F_q$-linear map such that $\tilde y(g_j)=y_j$ for $j=1,\dots,n$, that is,
\[
\tilde y\Bigl(\sum_{j=1}^n c_jg_j\Bigr)=\sum_{j=1}^n c_jy_j\qquad(c_j\in\F_q).
\]
\end{definition}

\begin{lemma}[vectors as linear maps]\label{lem:ytilde}
The map $y\mapsto\tilde y$ is a bijection between $\F_{q^m}^n$ and the set of $\F_q$-linear maps $G\to\F_{q^m}$, whose inverse is $\tilde y\mapsto(\tilde y(g_1),\dots,\tilde y(g_n))$. Under this bijection, the codeword $\lambda g$ corresponds to the multiplication by $\lambda$ on $G$, and $y-\lambda g$ corresponds to $\tilde y-\lambda$.
\end{lemma}

\begin{proof}
An $\F_q$-linear map on $G$ is determined by the images of the basis $g$, and these images can be prescribed arbitrarily. The map attached to $\lambda g$ sends $g_j$ to $\lambda g_j$, hence is $x\mapsto\lambda x$ by linearity; the last assertion follows since $y\mapsto\tilde y$ is $\F_q$-linear.
\end{proof}

In particular, choosing a received word is the same as choosing an $\F_q$-linear map on the support; this is how received words will be constructed in Section~\ref{sec:gab}.

\begin{lemma}\label{lem:rank}
$\rk(y-\lambda g)=n-\dim_{\F_q}\Ker(\tilde y-\lambda)$.
\end{lemma}

\begin{proof}
The coordinates of $y-\lambda g$ are the images $(\tilde y-\lambda)(g_j)$ of a basis of $G$, so they span $\mathrm{Im}(\tilde y-\lambda)$; conclude with the rank theorem.
\end{proof}

\begin{lemma}[eigenplane]\label{lem:plane}
There exists $\lambda$ with $\rk(y-\lambda g)\le n-2$ if and only if there exist $\F_q$-linearly independent $x_1,x_2\in G$ with $\tilde y(x_1)\,x_2=\tilde y(x_2)\,x_1$.
\end{lemma}

\begin{proof}
If $\dim\Ker(\tilde y-\lambda)\ge 2$, two independent vectors $x_1,x_2$ of this kernel satisfy $\tilde y(x_i)=\lambda x_i$, whence $\tilde y(x_1)x_2=\lambda x_1x_2=\tilde y(x_2)x_1$. Conversely put $\lambda=\tilde y(x_1)/x_1$; the equation gives $\tilde y(x_2)=\lambda x_2$, and since $\tilde y-\lambda$ is $\F_q$-linear the whole plane $\langle x_1,x_2\rangle$ lies in its kernel.
\end{proof}

For a single vector the condition is empty: any $x\ne0$ gives $\lambda=\tilde y(x)/x$, which is why radius $n-1$ is trivial. In Lemma~\ref{lem:plane} the equation is written for one basis $(x_1,x_2)$ of the plane; in Section~\ref{sec:gab} we shall apply it to other pairs of vectors of the same plane, which the following lemma allows.

\begin{lemma}[eigenplane equation for all pairs]\label{lem:allpairs}
Let $L\subset G$ be an $\F_q$-plane. The following are equivalent: (i) $\tilde y(x_1)x_2=\tilde y(x_2)x_1$ for one basis $(x_1,x_2)$ of $L$; (ii) $L\subset\Ker(\tilde y-\lambda)$ for some $\lambda\in\F_{q^m}$; (iii) $\tilde y(x)\,x'=\tilde y(x')\,x$ for all $x,x'\in L$.
\end{lemma}

\begin{proof}
(i)$\Rightarrow$(ii) is the second half of the proof of Lemma~\ref{lem:plane}, with $\lambda=\tilde y(x_1)/x_1$. (ii)$\Rightarrow$(iii): every $x\in L$ satisfies $\tilde y(x)=\lambda x$, hence for $x,x'\in L$, $\tilde y(x)x'=\lambda xx'=\tilde y(x')x$. (iii)$\Rightarrow$(i) is a special case.
\end{proof}

We call a plane satisfying these conditions a \emph{solution plane}; by Lemma~\ref{lem:rank}, the decoding instance is a yes-instance if and only if a solution plane exists.

The eigenplane formulation tells us what has to be encoded: a plane on which $\tilde y$ acts as a scalar. The reduction will be organized around an intermediate problem tailored to this, which we introduce now. Let $E=\F_{q^N}$. For $a\in E$ we call $a^{q+1}=a\cdot a^q$ the \emph{twisted square} of $a$: it is a square twisted by the Frobenius, and for $a\in\F_{q^2}$ it would be the norm to $\F_q$, but here $N$ will be odd.

\begin{definition}[Twisted Square Membership, $\Piv$]
Given an integer $N\equiv\pm1\pmod 6$ and two $\F_q$-subspaces $U,W$ of $E=\F_{q^N}$ by bases, decide whether there exists $a\in U\setminus\{0\}$ whose twisted square $a^{q+1}$ lies in $W$.
\end{definition}

The problem is in NP, with $a$ as certificate. Since $a\mapsto a^q$ is $\F_q$-linear, the twisted square is a quadratic form in the coordinates of $a$, which is what allows $\Piv$ to encode $\HomMQ$ (Section~\ref{sec:mq}); and it is the non-linear operation computed by an eigenplane equation, which is what allows $\Piv$ to embed into $\GabDec$ (Section~\ref{sec:gab}). The condition on $N$ combines two constraints of the latter step: $N$ odd, so that $\F_{q^2}\cap\F_{q^N}=\F_q$, and $N$ prime to $3$, so that a cubic extension of $\F_q$ stays cubic over $\F_{q^N}$. We are not aware of $\Piv$ having been studied before.

\section{Main result and overview of the reduction}\label{sec:overview}

We can now state the main result. It concerns the simplest Gabidulin codes, those of dimension one, at the first radius beyond which decoding is not trivial, $n-2$; in particular it does not depend on any special structure of the received word, which is part of the input. The field size $q$ is fixed, while the length $n$ and the extension degree $m$ are part of the input and are polynomial in the size of the source formula.

\begin{theorem}\label{thm:main}
For every fixed prime power $q$, $\GabDec$ is NP-complete. More precisely, a 3-SAT formula with $v$ variables and $c$ clauses is mapped in deterministic polynomial time to an instance with $m=3N$ and $n\le N+d$, where $d=1+v+c$, $N\equiv\pm1\pmod 6$ and $N=O(q\,d^2)$.
\end{theorem}

Membership in NP was noted in Section~\ref{sec:prelim}. NP-hardness is proved by a reduction from 3-SAT in three steps,
\[
\textsf{3-SAT}\ \xrightarrow{\ \S\ref{sec:sat}\ }\ \HomMQ\ \xrightarrow{\ \S\ref{sec:mq}\ }\ \Piv\ \xrightarrow{\ \S\ref{sec:gab}\ }\ \GabDec ,
\]
each step being an equivalence between a yes-instance of one problem and a yes-instance of the next (Propositions~\ref{prop:sat}, \ref{prop:mq} and~\ref{prop:gab}). Before giving the details, we describe the three steps once in a row, following the same object along the chain: a satisfying assignment of the formula becomes a non-zero solution $x$ of a quadratic system, then a field element $a$, and finally a plane $\langle a,\tau a^{q+1}\rangle$ on which the received word acts as a scalar.

\paragraph{From $\textsf{3-SAT}$ to $\HomMQ$.} Given a formula $\Phi$ with variables $b_1,\dots,b_v$ and clauses $C$, we write a system of quadratic equations in $d=1+v+c$ unknowns over $\F_q$: an unknown $h$ which stands for the constant $1$, one unknown $z_i$ per variable, which stands for its truth value, and one auxiliary unknown $w_C$ per clause. The equations say that $z_i\in\{0,1\}$ and that at least one literal of each clause is true. The unknown $h$ is there to make all equations homogeneous: a satisfying assignment gives a solution with $h=1$, and conversely homogeneity and the equations $z_i^2=hz_i$ force any non-zero solution to have $h\ne0$, hence to come from a satisfying assignment. This yields homogeneous quadratic forms $Q_1,\dots,Q_s$ such that $\Phi$ is satisfiable if and only if the $Q_k$ have a common zero $x\ne0$.

\paragraph{From $\HomMQ$ to $\Piv$.} We encode the $d$ coordinates of $x$ into a single field element so that all the quadratic forms become linear conditions on its twisted square. We take $E=\F_{q^N}$ with a generator $\alpha$, choose exponents $e_1<\dots<e_d$, and encode $x$ as $a=\phi(x)=x_1\alpha^{e_1}+\dots+x_d\alpha^{e_d}$; this identifies $\F_q^d$ with the subspace $U=\langle\alpha^{e_1},\dots,\alpha^{e_d}\rangle$, and $x\ne0$ with $a\ne0$. Expanding $a^{q+1}=a\cdot a^q$ gives $\sum_{i\le j}x_ix_j\mu_{ij}$ with $\mu_{ij}\in E$; the exponents are chosen so that the $\mu_{ij}$ are linearly independent, which means that $a^{q+1}$ "remembers" every monomial $x_ix_j$ separately. Trace duality then provides, for each $k$, an element $t_k\in E$ with $\Tr(t_ka^{q+1})=Q_k(x)$ for all $x$. Putting $W=\{z\in E:\Tr(t_kz)=0\text{ for all }k\}$, the conditions $Q_k(x)=0$ become $a^{q+1}\in W$. This gives an instance $(U,W)$ of $\Piv$ that is a yes-instance if and only if the system has a non-zero solution.

\paragraph{From $\Piv$ to $\GabDec$.} From $(U,W)$ we build a code and a received word, i.e.\ a support $G$ and an $\F_q$-linear map $\tilde y$ on it, so that the eigenplanes of $\tilde y$ are exactly the planes $\langle a,\tau a^{q+1}\rangle$ with $a$ a solution of $\Piv$. Three things must be imposed on a plane: its shape, and the two conditions of $\Piv$. The linear conditions are placed in the support: $G=U\oplus\tau W$ inside a cubic extension $K=E(\tau)$, so that every vector of $G$ is $a+\tau b$ with $a\in U$, $b\in W$. The eigenplane equation, one equality in $K$, is then three equalities in $E$, one per coefficient on $1,\tau,\tau^2$, and this is what the cubic extension is for. We take $\tilde y(a+\tau b)=a^{q^2}+\tau b^q$, two Frobenius twists with different exponents. Two of the three equalities force the plane to be $\langle a,\tau b\rangle$; the third compares the two twists and forces $b$ to be the twisted square $a^{q+1}$, up to a scalar of $\F_q$, so that $a^{q+1}\in W$. Conversely a solution $a$ of $\Piv$ gives the eigenplane $\langle a,\tau a^{q+1}\rangle$. Thus $(g,y)$ is a yes-instance exactly when $(U,W)$ is; the details are in Section~\ref{sec:gab}.

The sizes are summarized in Table~\ref{tab:sizes} at the end of Section~\ref{sec:gab}. The three steps are now treated in turn, each in a self-contained section.

\section{From $\textsf{3-SAT}$ to $\HomMQ$}\label{sec:sat}

This first step is classical, and we include it because its two features, homogeneity and the special unknown $h$, are used later. We start from a 3-SAT instance as in Section~\ref{sec:prelim}, with variables $b_1,\dots,b_v$ and $c$ clauses $C=\ell_1\vee\ell_2\vee\ell_3$. We may assume $v+c\ge1$, so that $d=1+v+c\ge2$; the trivial formula with no variable and no clause is satisfiable and is mapped to a fixed yes-instance. From such an instance we construct an instance of $\HomMQ$: homogeneous quadratic forms $Q_1,\dots,Q_s$ in $d$ unknowns over $\F_q$ such that the formula is satisfiable if and only if the forms have a common zero $x\ne0$.

The translation is the natural one. A truth value becomes an element of $\{0,1\}\subset\F_q$, encoded by the equation $z^2=z$. A literal is represented by a linear expression which vanishes exactly when the literal is true, and a clause is true when the product of the three expressions of its literals vanishes; since this product has degree three, an auxiliary unknown is introduced per clause to split it into equations of degree two. Finally, the constant $1$ is replaced by an unknown $h$, for the reason explained below.

\paragraph{The role of the unknown $h$ for homogenization.} We recall here the standard technique of homogenization \cite{CLO15}, and explain why it is needed. The unknown $h$ is imposed by the pivot problem $\Piv$ (Twisted Square Membership) to which the quadratic system will be reduced in Section~\ref{sec:mq}. The twisted square $a^{q+1}=a\cdot a^q$ is a homogeneous quadratic form in the coordinates of $a$: every monomial is a product $x_ix_j$, with no constant and no linear term. Through trace duality, only homogeneous equations of degree two can therefore be encoded into $\Piv$. The natural equations are not homogeneous: a Boolean value $z$ satisfies $z^2-z=0$, which has a linear term, and a clause involves the constant $1$. We therefore homogenize: the constant $1$ is replaced by a new unknown $h$ and every linear term $z$ by $hz$, which turns an equation $f(z)=0$ of degree at most two into the homogeneous equation $h^2f(z/h)=0$; geometrically, this passes from an affine variety to its projective closure \cite[Ch.~8]{CLO15}. Its effect on solutions is the following. If $h\ne0$, every equation can be divided by $h^2$, and the solution divided by $h$ is a solution of the original system with $h=1$. If $h=0$, the equations $z_i^2-hz_i=0$ become $z_i^2=0$, so $z_i=0$, and likewise for the auxiliary unknowns, so that the solution is zero. Thus the homogenization creates no spurious solution, provided one asks for a \emph{non-zero} solution, which is exactly what $\HomMQ$ does.

For instance, consider the single clause $b_1\vee\neg b_2\vee b_3$. Without $h$, it is encoded by the equations $z_i^2-z_i=0$ and $(1-z_1)z_2(1-z_3)=0$, which cannot be fed to $\Piv$. With $h$, the equations become $z_i^2-hz_i=0$, $w^2-hw=0$, $hw-(h-z_1)z_2=0$ and $w(h-z_3)=0$, all homogeneous of degree two. The assignment $b_1=\text{true}$, $b_2=b_3=\text{false}$ gives the solution $(h,z_1,z_2,z_3,w)=(1,1,0,0,0)$. Over $\F_3$, the vector $(2,2,0,0,0)=2\cdot(1,1,0,0,0)$ is also a solution; it describes the same assignment, recovered by dividing by $h=2$. Setting $h=0$ forces $z_1=z_2=z_3=w=0$.

Concretely, the unknowns are $h$, one unknown $z_i$ per variable $b_i$, and one auxiliary unknown $w_C$ per clause $C$, so that $d=1+v+c$. To a literal we associate the linear form
\[
u(b_i)=h-z_i,\qquad u(\neg b_i)=z_i ,
\]
which, for $h=1$, vanishes exactly when the literal is true. For a clause $C=\ell_1\vee\ell_2\vee\ell_3$ we write $u_j=u(\ell_j)$. The system consists of the following equations, $v$ of the first kind and $3c$ of the other three, so that $s=v+3c$:
\[
z_i^2-hz_i=0,\qquad w_C^2-hw_C=0,\qquad hw_C-u_1u_2=0,\qquad w_Cu_3=0 .
\]
The first family says that $z_i$ is a truth value, the third defines $w_C$ as $u_1u_2$, and the last one then says $u_1u_2u_3=0$, i.e.\ that one literal of $C$ is true; the second family only normalizes $w_C$.

Before proving that this system does what we want, let us see it on a small case. As a running example, take the formula consisting of the single clause $C=b_1\vee\neg b_2\vee b_3$, so that $v=3$, $c=1$ and $d=5$. The unknowns are $h,z_1,z_2,z_3,w_C$, the literal forms are $u_1=h-z_1$, $u_2=z_2$, $u_3=h-z_3$, and the system consists of the $s=v+3c=6$ equations
\[
z_1^2-hz_1,\quad z_2^2-hz_2,\quad z_3^2-hz_3,\quad w_C^2-hw_C,\quad hw_C-(h-z_1)z_2,\quad w_C(h-z_3).
\]
The assignment $b_1=\text{true}$, $b_2=b_3=\text{false}$ gives the solution $(h,z_1,z_2,z_3,w_C)=(1,1,0,0,0)$, and indeed $u_1=0$. In the notation of Section~\ref{sec:mq}, where the unknowns are numbered $x_1=h$, $x_2=z_1$, $x_3=z_2$, $x_4=z_3$, $x_5=w_C$, the fifth equation reads $x_1x_5-x_1x_3+x_2x_3$, with coefficients $c_{15}=1$, $c_{13}=-1$, $c_{23}=1$ and all others zero.

It remains to verify that satisfying assignments correspond exactly to non-zero solutions. The only point requiring care is that a non-zero solution must have $h\ne0$.

\begin{proposition}\label{prop:sat}
The formula is satisfiable if and only if the system has a non-zero solution in $\F_q^d$.
\end{proposition}

\begin{proof}
Given a satisfying assignment, set $h=1$, $z_i\in\{0,1\}$ its value and $w_C=u_1u_2$; the last equation holds since one literal of $C$ is true, so $u_1u_2u_3=0$. Conversely let $x$ be a non-zero solution. If $h=0$ the first two equations give $z_i=w_C=0$, hence $x=0$; so $h\ne0$, and by homogeneity we may take $h=1$. Then $z_i\in\{0,1\}$ defines an assignment, $w_C=u_1u_2$, and $u_1u_2u_3=0$ says that a literal of each clause is true.
\end{proof}

The output of this step, the list $Q_1,\dots,Q_s$ of the left-hand sides above, read as forms in the $d$ unknowns $h,z_1,\dots,z_v,(w_C)_C$, is the input of Section~\ref{sec:mq}.

\section{From $\HomMQ$ to $\Piv$}\label{sec:mq}

We start from an instance of $\HomMQ$, that is, homogeneous quadratic forms $Q_1,\dots,Q_s$ in $d$ unknowns over $\F_q$, for instance those produced in Section~\ref{sec:sat}, and we construct an instance $(N,U,W)$ of $\Piv$, that is, a field $E=\F_{q^N}$ with $N\equiv\pm1\pmod 6$ and two subspaces $U,W\subset E$ given by bases, such that the forms have a common zero $x\ne0$ if and only if some $a\in U\setminus\{0\}$ has $a^{q+1}\in W$ (Proposition~\ref{prop:mq}). The plan is to encode the unknown vector $x\in\F_q^d$ by a single element $a$ of a field $E=\F_{q^N}$, chosen so that every form $Q_k(x)$ can be read off $a^{q+1}$ by a linear operation. Two tools make this possible. Trace duality (Section~\ref{sec:trace}) expresses every $\F_q$-linear form on $E$ as $z\mapsto\Tr(tz)$ and lets us prescribe the values of such a form on independent elements. An explicit choice of exponents (Section~\ref{sec:expo}) ensures that, when $a=\sum_ix_i\alpha^{e_i}$ is raised to the power $q+1$, each monomial $x_ix_j$ comes with its own independent coefficient in $E$, so that the forms can be prescribed freely. The encoding itself and the proof of the equivalence are in Section~\ref{sec:encoding}.

\subsection{Trace duality}\label{sec:trace}

Let $E=\F_{q^N}$ and $\Tr(z)=z+z^q+\dots+z^{q^{N-1}}$, an $\F_q$-linear map $E\to\F_q$ which is not identically zero, being a polynomial of degree $q^{N-1}<|E|$.

\begin{lemma}\label{lem:trace}
Every $\F_q$-linear form on $E$ is $z\mapsto\Tr(tz)$ for a unique $t\in E$. Consequently, if $\mu_1,\dots,\mu_r\in E$ are $\F_q$-linearly independent and $c_1,\dots,c_r\in\F_q$, there exists $t\in E$ with $\Tr(t\mu_j)=c_j$ for all $j$, computable by linear algebra.
\end{lemma}

\begin{proof}
If $t\ne0$, $z\mapsto tz$ is a bijection of $E$, so $z\mapsto\Tr(tz)$ is not zero. Hence $t\mapsto\Tr(t\,\cdot)$ is an injective linear map between two spaces of dimension $N$, thus bijective. For the second claim, complete the $\mu_j$ into a basis and prescribe a linear form on it.
\end{proof}

\subsection{Choice of exponents}\label{sec:expo}

Let us first see what is needed. The unknown $x\in\F_q^d$ will be encoded as $a=\sum_{i=1}^dx_i\alpha^{e_i}$, where $\alpha$ generates $E$ over $\F_q$ and $e_1<\dots<e_d$ are integers to be chosen. Since $x_i^q=x_i$ for $x_i\in\F_q$, the twisted square expands as
\[
a^{q+1}=a\cdot a^q=\Bigl(\sum_ix_i\alpha^{e_i}\Bigr)\Bigl(\sum_jx_j\alpha^{qe_j}\Bigr)=\sum_{i,j}x_ix_j\,\alpha^{e_i+qe_j}.
\]
So the twisted square is the sum, over all ordered pairs $(i,j)$, of the monomial $x_ix_j$ times the power $\alpha^{e_i+qe_j}$. We want to recover each monomial from $a^{q+1}$ by a linear operation, and this is possible as soon as the powers $\alpha^{e_i+qe_j}$ are linearly independent over $\F_q$, which holds when the exponents $e_i+qe_j$ are pairwise distinct and smaller than $N$, since $1,\alpha,\dots,\alpha^{N-1}$ is a basis. The condition on the exponents is thus
\[
(\star)\qquad\text{the } d^2 \text{ integers } e_i+qe_j\ (1\le i,j\le d)\ \text{ are pairwise distinct,}
\]
the case $i=j$ giving $(q+1)e_i$; and the size of the largest of them will dictate $N$, hence the size of the whole construction. We therefore want integers satisfying $(\star)$ which are as small as possible.

The idea of the construction is that $e_i+qe_j$ should display $i$ and $j$ side by side, without interference. Write $i-1$ in base $q$ and dilate its expansion, placing its coefficients on the even positions; the odd positions are then free, and multiplying by $q$ shifts a number by one position. Thus $e_i$ occupies the even positions and $qe_j$ the odd ones, and the sum $e_i+qe_j$ is simply the interleaving of the two base-$q$ expansions, from which $i$ and $j$ are read off. Formally, for an integer $t\ge0$ with base-$q$ expansion $t=\sum_kt_kq^k$, with coefficients $0\le t_k<q$, put
\[
\delta_q(t)=\sum_kt_kq^{2k},
\]
and set $e_i=\delta_q(i-1)$. For $q=2$, $t=5=\mathtt{101}_2$ gives $\delta_2(5)=\mathtt{10001}_2=17$, and the first exponents are $0,1,4,5,16$.

\begin{lemma}\label{lem:expo}
The integers $e_i=\delta_q(i-1)$ are strictly increasing, satisfy $(\star)$, and $e_d\le(d-1)^2$. Hence all integers of $(\star)$ are at most $(q+1)(d-1)^2$.
\end{lemma}

\begin{proof}
$\delta_q$ is strictly increasing, since it preserves the lexicographic comparison of coefficients. Let $(t_k)$ and $(s_k)$ be the base-$q$ coefficients of $i-1$ and $j-1$. Then $e_i+qe_j=\sum_kt_kq^{2k}+\sum_ks_kq^{2k+1}$: the coefficients of $e_i$ sit on even positions, those of $qe_j$ on odd positions, and all are smaller than $q$, so the two expansions interleave and this is the base-$q$ expansion of the sum, with coefficients $t_0,s_0,t_1,s_1,\dots$ By uniqueness of the expansion, the even positions give back $i-1$ and the odd positions $j-1$, so the ordered pair $(i,j)$ is recovered and $(\star)$ holds. For instance, for $q=2$ and $(i,j)=(2,3)$, $e_2+2e_3=1+8=9=\mathtt{1001}_2$, whose even positions read $1=i-1$ and odd positions read $2=j-1$. Finally $t_k^2\ge t_k$ gives $\delta_q(t)\le\sum_kt_k^2q^{2k}\le(\sum_kt_kq^k)^2=t^2$, whence $e_d\le(d-1)^2$.
\end{proof}

\subsection{The encoding and the reduction}\label{sec:encoding}

We now fix the field and encode the unknown. Let $N$ be the smallest integer with $N\equiv\pm1\pmod 6$ and $N>(q+1)(d-1)^2$, so $N=O(qd^2)$; the congruence will be used in Section~\ref{sec:gab}. We build $E=\F_q[X]/(P)$ with $P$ irreducible of degree $N$, found in deterministic polynomial time by Shoup's algorithm \cite{Sho90}, and let $\alpha$ be the class of $X$. Since $P$ is the minimal polynomial of $\alpha$, the powers $1,\alpha,\dots,\alpha^{N-1}$ form an $\F_q$-basis of $E$: a non-trivial relation would give a non-zero polynomial of degree less than $N$ vanishing at $\alpha$.

The quadratic-system unknown is encoded as follows. The unknown is a single vector $x=(x_1,\dots,x_d)\in\F_q^d$: its coordinate $x_i\in\F_q$ is the value given to the $i$-th unknown of the system, with the numbering $x_1=h$, $x_2,\dots,x_{v+1}$ the values of $z_1,\dots,z_v$ and the last $c$ coordinates the values of the $w_C$, as in the example of Section~\ref{sec:sat}; $Q_k(x)\in\F_q$ is the value of the $k$-th equation at this point, and we look for $x\ne0$ with $Q_k(x)=0$ for all $k$. The aim of this step is to replace the $d$ coordinates of $x$ by one element of the field $E$. Let $U=\langle\alpha^{e_1},\dots,\alpha^{e_d}\rangle_{\F_q}$, a subspace of dimension $d$ since the $\alpha^{e_i}$ are distinct elements of a basis, and define
\[
\phi:\F_q^d\longrightarrow U,\qquad x=(x_1,\dots,x_d)\longmapsto a=x_1\alpha^{e_1}+\dots+x_d\alpha^{e_d}.
\]
Thus the $x_i$ are the coordinates of $a$ in the basis $(\alpha^{e_i})$ of $U$, and $\phi$ is an $\F_q$-linear bijection: looking for $x$ or looking for $a\in U$ is the same thing, and $a\ne0$ if and only if $x\ne0$. For $q=2$ and $d=2$ one has $a=x_1+x_2\alpha$, and $a^3=a\cdot a^2=(x_1+x_2\alpha)(x_1+x_2\alpha^2)=x_1^2+x_1x_2(\alpha+\alpha^2)+x_2^2\alpha^3$: each monomial of the system comes out with its own coefficient in $E$, which is the point of the encoding. In general, since $x_i^q=x_i$,
\[
a^{q+1}=a\cdot a^q=\sum_{i\le j}x_ix_j\,\mu_{ij},\qquad \mu_{ii}=\alpha^{(q+1)e_i},\quad \mu_{ij}=\alpha^{e_i+qe_j}+\alpha^{e_j+qe_i}\ (i<j).
\]
By Lemma~\ref{lem:expo} and $N>(q+1)(d-1)^2$, the elements $\mu_{ij}$ involve pairwise disjoint sets of basis vectors $\alpha^e$, $e<N$, and are therefore linearly independent over $\F_q$.

Write $Q_k(x)=\sum_{i\le j}c^{(k)}_{ij}x_ix_j$, where $c^{(k)}_{ij}\in\F_q$ is the coefficient of the monomial $x_ix_j$ in the $k$-th equation; the coordinates of $x$ are the values of $h$, of the $z_i$ and of the $w_C$, in this order, as in the example of Section~\ref{sec:sat}. For instance the equation $z_1^2-hz_1$ reads $x_2^2-x_1x_2$, so that $c_{22}=1$, $c_{12}=-1$ and all other coefficients vanish. Since the $\mu_{ij}$ are independent, Lemma~\ref{lem:trace} lets us prescribe a linear form on them: we choose $t_k\in E$ with $\Tr(t_k\mu_{ij})=c^{(k)}_{ij}$ for all $i\le j$. Since $x_ix_j\in\F_q$ and the trace is $\F_q$-linear,
\[
\Tr(t_k\,a^{q+1})=\sum_{i\le j}x_ix_j\Tr(t_k\mu_{ij})=Q_k(x)\qquad\text{for all }x\in\F_q^d .
\]
Finally put $W=\{z\in E:\Tr(t_kz)=0,\ k=1,\dots,s\}$, a subspace computed by linear algebra.

\begin{proposition}\label{prop:mq}
The quadratic system has a non-zero solution if and only if $(U,W)$ is a yes-instance of $\Piv$.
\end{proposition}

\begin{proof}
Let $x\in\F_q^d$ and $a=\phi(x)\in U$. Since $\phi$ is a linear bijection, $a\ne0$ if and only if $x\ne0$. By definition of $W$, $a^{q+1}\in W$ if and only if $\Tr(t_ka^{q+1})=0$ for all $k$, and by the choice of the $t_k$ this reads $Q_k(x)=0$ for all $k$. Hence if $x\ne0$ is a common zero of the $Q_k$, then $a=\phi(x)$ is a non-zero element of $U$ with $a^{q+1}\in W$; conversely, if $a\in U\setminus\{0\}$ satisfies $a^{q+1}\in W$, then $x=\phi^{-1}(a)$ is non-zero and $Q_k(x)=\Tr(t_ka^{q+1})=0$ for all $k$.
\end{proof}

\begin{corollary}\label{cor:tsm}
For every fixed prime power $q$, $\Piv$ is NP-complete under deterministic polynomial-time reductions.
\end{corollary}

\begin{proof}
$\Piv$ is in NP: a certificate is $a$, and the conditions $a\in U\setminus\{0\}$ and $a^{q+1}\in W$ are checked by linear algebra over $\F_q$. NP-hardness follows from Propositions~\ref{prop:sat} and~\ref{prop:mq}, whose composition is a deterministic polynomial-time reduction from 3-SAT to $\Piv$.
\end{proof}

This completes the second step. Its output is the instance $(N,U,W)$ of $\Piv$, with $U$ given by the basis $\alpha^{e_1},\dots,\alpha^{e_d}$ and $W$ by a basis computed from the $t_k$; all the non-linearity of the quadratic system is now carried by the single operation $a\mapsto a^{q+1}$, and the two remaining conditions, $a\in U$ and $a^{q+1}\in W$, are linear. The next section shows how a Gabidulin decoding instance can impose exactly these three things.

\section{From $\Piv$ to $\GabDec$}\label{sec:gab}

\subsection{What has to be shown}

Section~\ref{sec:mq} turned a homogeneous quadratic system into an instance $(N,U,W)$ of $\Piv$: a field $E=\F_{q^N}$ with $N\equiv\pm1\pmod 6$, and two subspaces $U,W\subset E$ given by bases, such that the system has a non-zero solution if and only if some $a\in U\setminus\{0\}$ has $a^{q+1}\in W$. All the non-linearity of the original problem is now carried by the twisted square $a\mapsto a^{q+1}$, while the conditions $a\in U$ and $a^{q+1}\in W$ are linear. In this last step we start from such an instance and construct an instance $(g,y)$ of $\GabDec$, that is, a field $K=\F_{q^m}$, a vector $g\in K^n$ of rank $n$ and a received word $y\in K^n$, such that some $\lambda\in K$ satisfies $\rk(y-\lambda g)\le n-2$ if and only if some $a\in U\setminus\{0\}$ has $a^{q+1}\in W$. The received word is not given, it is designed.

By Lemma~\ref{lem:plane}, the decoding instance is a yes-instance exactly when some plane of the support satisfies the eigenplane equation. One direction consists in exhibiting such a plane from a solution of $\Piv$. The converse is the heart of the proof: we show that every solution plane is necessarily of the form $\langle a,\tau a^{q+1}\rangle$ with $a\in U\setminus\{0\}$ and $a^{q+1}\in W$. The construction and the proof are organized as follows. In Section~\ref{sec:K} we build a cubic extension $K=E(\tau)$ of $E$, in which one equality splits into three independent equalities in $E$, one per coefficient on $1,\tau,\tau^2$; this is what allows a single eigenplane equation to impose several conditions. In Section~\ref{sec:code} we choose the code, through its support $G=U\oplus\tau W$, so that the two linear conditions of $\Piv$ are built into the code itself. In Section~\ref{sec:word} we choose the received word, $\tilde y(a+\tau b)=a^{q^2}+\tau b^q$, made of two Frobenius twists with different exponents, and check that it is a legitimate received word. Finally, in Section~\ref{sec:shape}, we read the eigenplane equation coefficient by coefficient: the coefficients on $1$ and $\tau^2$ force the shape $\langle a,\tau b\rangle$, the coefficient on $\tau$ forces $b\in\F_q^*a^{q+1}$, and this proves the equivalence (Proposition~\ref{prop:gab}) and the main theorem.

\subsection{The field $K$}\label{sec:K}

We first build the field in which the code will live. Fix once and for all an irreducible polynomial $R\in\F_q[Y]$ of degree $3$; it depends only on $q$, and it exists since the field $\F_{q^3}$ exists and the minimal polynomial over $\F_q$ of a generator of its cyclic multiplicative group is irreducible of degree $3$ (for $q=2$, $Y^3+Y+1$); it is found by exhaustive search in constant time. Put $K=E[Y]/(R)$ and let $\tau$ be the class of $Y$.

\begin{lemma}\label{lem:K}
$K$ is a field, $K=\F_{q^{3N}}$, and $1,\tau,\tau^2$ is a basis of $K$ over $E$.
\end{lemma}

\begin{proof}
An irreducible polynomial of degree $3$ over $\F_q$ factors over $\F_{q^N}$ into $\gcd(3,N)$ irreducible factors of equal degree, and $\gcd(3,N)=1$ since $N\equiv\pm1\pmod 6$. Hence $R$ stays irreducible over $E$, $[K:E]=3$, and $1,\tau,\tau^2$ is a basis by the argument used for $\alpha$.
\end{proof}

Every $z\in K$ is thus written uniquely $z=z_0+z_1\tau+z_2\tau^2$ with $z_i\in E$, and an equality in $K$ is equivalent to three equalities in $E$, one for each power of $\tau$; we call $z_i$ the component of $z$ on $\tau^i$.

\subsection{The code}\label{sec:code}

With $K$ in hand we choose the code, that is, its support. Recall from Section~\ref{sec:prelim} that a code $\Gab(g,1)$ is determined by the vector $g=(g_1,\dots,g_n)$, whose coordinates are $\F_q$-linearly independent, and that its \emph{support} is the $\F_q$-space $G=\langle g_1,\dots,g_n\rangle_{\F_q}\subset\F_{q^m}$ of dimension $n$ spanned by these coordinates: $g$ is a basis of $G$, a received word is an $\F_q$-linear map $\tilde y:G\to\F_{q^m}$, and decoding asks for a plane of $G$ on which $\tilde y$ acts as a scalar. Choosing the code thus amounts to choosing $G$ together with a basis. Since the expected plane contains a vector $a\in U$ and a vector $\tau b$ with $b\in W$, we take
\[
G=U\oplus\tau W=\{a+\tau b:a\in U,\ b\in W\}\subset K .
\]
The sum is direct, because $1$ and $\tau$ are linearly independent over $E$, so, writing $r=\dim W$ (with $N-s\le r\le N$, since $W$ is cut out by $s$ linear conditions), $n:=\dim_{\F_q}G=\dim U+\dim W=d+r\le d+N$. We choose the extension field $\F_{q^m}=K$ with $m=3N$, and for $g\in K^n$ the concatenation of the basis $\alpha^{e_1},\dots,\alpha^{e_d}$ of $U$ and of $\tau w_1,\dots,\tau w_r$ for a basis $w_1,\dots,w_r$ of $W$ (computed by linear algebra from the $t_k$); this is an $\F_q$-basis of $G$, so $\rk(g)=n$. Note that $n\ge d\ge2$, so that radius $n-2$ makes sense. Moreover $N>(q+1)(d-1)^2\ge d$ since $d\ge2$, hence
\[
n\le d+N\le 2N<3N=m,
\]
so that the code satisfies the length constraint $n\le m$ of Gabidulin codes.

\subsection{The received word}\label{sec:word}

The code being fixed, it remains to choose the received word; this is where the twisted square is produced. Since $(g,y)$ is ours to choose, and since by Section~\ref{sec:prelim} received words correspond exactly to $\F_q$-linear maps $G\to K$, we define the map $\tilde y$ and deduce $y$ from it. Every element of $G$ is uniquely $a+\tau b$ with $a\in U$, $b\in W$, so it suffices to give $\tilde y(a+\tau b)$. We set
\[
(\ast)\qquad \tilde y(a+\tau b)=a^{q^2}+\tau\,b^{q}.
\]

\begin{lemma}\label{lem:word}
The map $(\ast)$ is $\F_q$-linear. The corresponding received word $y_j=\tilde y(g_j)$ is $y_j=\alpha^{e_iq^2}$ for $g_j=\alpha^{e_i}$ and $y_j=\tau w_l^q$ for $g_j=\tau w_l$, and is computed in polynomial time.
\end{lemma}

\begin{proof}
The map $a+\tau b\mapsto(a,b)$ is linear; $a\mapsto a^{q^2}$ and $b\mapsto b^q$ are additive and satisfy $(\lambda a)^{q^2}=\lambda a^{q^2}$ for $\lambda\in\F_q$; the multiplication by $\tau$ is linear.
\end{proof}

The two Frobenius twists have distinct exponents on purpose: equal exponents on $a$ and $b$ would force the plane to be $\langle a,\tau b\rangle$ as well, but would only give $b\in\F_q^*a$, whereas the exponents $q^2$ and $q$ produce $b\in\F_q^*a^{q+1}$. For $x=a+\tau b$ and $x'=a'+\tau b'$ in $G$, the components of $\tilde y(x)x'$ on $1,\tau,\tau^2$ are
\[
a^{q^2}a',\qquad a^{q^2}b'+b^qa',\qquad b^qb' ,
\]
of degree at most $2$ in $\tau$, so no reduction modulo $R$ occurs.

\subsection{The shape of solution planes}\label{sec:shape}

We prove the converse direction by showing that every solution plane has the form $\langle a,\tau b\rangle$. The first lemma is the elementary fact that makes the coefficients on $1$ and $\tau^2$ useful.

\begin{lemma}\label{lem:Fq}
Let $a,a'\in E$ be non-zero with $a^{q^2}a'=a'^{q^2}a$. Then $a'\in\F_q\,a$.
\end{lemma}

\begin{proof}
Put $u=a'/a$. Dividing by $a^{q^2}a$ gives $u^{q^2}=u$. Since $u\in\F_{q^N}$, also $u^{q^N}=u$. Write $N=2k+1$; applying $k$ times $u\mapsto u^{q^2}$ gives $u^{q^{2k}}=u$, hence $u=u^{q^N}=(u^{q^{2k}})^q=u^q$, and $u\in\F_q$.
\end{proof}

This is the only place where the parity of $N$ is used.

\begin{lemma}\label{lem:shape}
If a plane $L\subset G$ satisfies $L\subset\Ker(\tilde y-\lambda)$ for some $\lambda\in K$, then $L=\langle a,\tau b\rangle$ with $a\in U\setminus\{0\}$ and $b\in W\setminus\{0\}$.
\end{lemma}

\begin{proof}
By Lemma~\ref{lem:allpairs}, for all $x=a+\tau b$ and $x'=a'+\tau b'$ in $L$ we have $\tilde y(x)x'=\tilde y(x')x$, and the components on $1$ and $\tau^2$ give $a^{q^2}a'=a'^{q^2}a$ and $b^qb'=b'^qb$. Let $A=\{a\in U:\exists b,\ a+\tau b\in L\}$, the image of $L$ under the linear map $a+\tau b\mapsto a$; it is a subspace of $U$. If $\dim A\ge2$, take $a,a'\in A$ linearly independent and $x,x'\in L$ realizing them: the first relation and Lemma~\ref{lem:Fq} give $a'\in\F_qa$, a contradiction. Hence $\dim A\le1$, and likewise $\dim B\le1$ for $B=\{b\in W:\exists a,\ a+\tau b\in L\}$, using $(b'/b)^q=b'/b$. The map $L\to A\times B$, $a+\tau b\mapsto(a,b)$, is linear and injective, so $2=\dim L\le\dim A+\dim B\le2$. Thus $A=\F_qa$, $B=\F_qb$ with $a,b\ne0$, the map is bijective, and $a$ and $\tau b$ belong to $L$.
\end{proof}

Note that the argument never assumes that a particular basis of $L$ has non-zero components: it works on the projections $A$ and $B$ of the whole plane. With the shape of solution planes in hand, the coefficient on $\tau$ finishes the job, and the converse direction is a direct computation.

\begin{proposition}\label{prop:gab}
There exists $\lambda\in K$ with $\rk(y-\lambda g)\le n-2$ if and only if there exists $a\in U\setminus\{0\}$ with $a^{q+1}\in W$.
\end{proposition}

\begin{proof}
Assume the decoding instance is a yes-instance. By Lemmas~\ref{lem:rank} and~\ref{lem:shape}, some plane $\langle a,\tau b\rangle$ with $a\in U\setminus\{0\}$ and $b\in W\setminus\{0\}$ lies in $\Ker(\tilde y-\lambda)$. The component on $\tau$ of the equality $\tilde y(a)\,\tau b=\tilde y(\tau b)\,a$ reads $a^{q^2}b=ab^q$. Dividing by $ab$ gives $a^{q^2-1}=b^{q-1}$, and since $a^{q^2-1}=(a^{q+1})^{q-1}$, the element $\gamma=b/a^{q+1}$ satisfies $\gamma^{q-1}=1$, so $\gamma\in\F_q^*$; this is where the exponent $q+1$ of $\Piv$ arises. Since $W$ is an $\F_q$-subspace, $a^{q+1}=\gamma^{-1}b\in W$.

Conversely let $a\in U\setminus\{0\}$ with $a^{q+1}\in W$, put $b=a^{q+1}$, so that $\tau b\in G$, and $\lambda=a^{q^2-1}$. Then $\tilde y(a)=a^{q^2}=\lambda a$ and $\tilde y(\tau b)=\tau a^{q^2+q}=\lambda\tau b$. By linearity $\langle a,\tau b\rangle$, a plane since the sum $U\oplus\tau W$ is direct, lies in $\Ker(\tilde y-\lambda)$, and Lemma~\ref{lem:rank} gives $\rk(y-\lambda g)\le n-2$.
\end{proof}

\subsection{Proof of Theorem~\ref{thm:main}}\label{sec:proofthm}

All three steps are now established, and we can prove the main theorem, which we recall.

\medskip\noindent{\bfseries Theorem~\ref{thm:main}.}\ \emph{For every fixed prime power $q$, $\GabDec$ is NP-complete. More precisely, a 3-SAT formula with $v$ variables and $c$ clauses is mapped in deterministic polynomial time to an instance with $m=3N$ and $n\le N+d$, where $d=1+v+c$, $N\equiv\pm1\pmod 6$ and $N=O(q\,d^2)$.}
\medskip

\begin{proof}
\emph{Membership in NP.} Given $(g,y)$, a certificate is $\lambda\in\F_{q^m}$; the rank of $y-\lambda g$ is computed by Gaussian elimination over $\F_q$ on its $m\times n$ matrix of coordinates, in polynomial time.

\emph{NP-hardness.} Let $\Phi$ be a 3-SAT formula with $v$ variables and $c$ clauses, and $d=1+v+c$. Section~\ref{sec:sat} produces homogeneous quadratic forms $Q_1,\dots,Q_s$ in $d$ unknowns, with $s=v+3c$, such that $\Phi$ is satisfiable if and only if they have a common zero $x\ne0$ (Proposition~\ref{prop:sat}). Section~\ref{sec:mq} produces from them an instance $(N,U,W)$ of $\Piv$, with $N\equiv\pm1\pmod 6$ and $N=O(qd^2)$, such that the forms have a common zero $x\ne0$ if and only if some $a\in U\setminus\{0\}$ has $a^{q+1}\in W$ (Proposition~\ref{prop:mq}). Section~\ref{sec:gab} produces from $(N,U,W)$ an instance $(g,y)$ of $\GabDec$ over $K=\F_{q^{3N}}$, of length $n=\dim U+\dim W\le d+N$, such that some $\lambda$ satisfies $\rk(y-\lambda g)\le n-2$ if and only if $(N,U,W)$ is a yes-instance (Proposition~\ref{prop:gab}). Composing the three equivalences, $\Phi$ is satisfiable if and only if $(g,y)$ is a yes-instance of $\GabDec$.

\emph{Polynomial time.} Each object of the construction is computed deterministically in time polynomial in $v+c$, with the sizes of Table~\ref{tab:sizes}: the exponents $e_i$ are given by a closed formula (Lemma~\ref{lem:expo}); the irreducible polynomial of degree $N$ defining $E$ is found by Shoup's algorithm \cite{Sho90}; the polynomial $R$ of degree $3$ defining $K$ depends only on $q$; the elements $t_k$, the basis of $W$, the vector $g$ and the received word $y$ are obtained by linear algebra over $\F_q$ in dimension $O(m)=O(qd^2)$ (Lemmas~\ref{lem:trace} and~\ref{lem:word}). Since $q$ is fixed, $m,n=O(d^2)$, and the $2n$ elements of $\F_{q^m}$ forming $g$ and $y$ have $m$ coordinates each over $\F_q$, so that the output instance has size $|(g,y)|=O(nm)=O(d^4)$ symbols of $\F_q$, polynomial in the size of the formula.
\end{proof}

\begin{remark}[representation of finite fields]
The answer does not depend on how finite fields are represented. We take $E=\F_q[X]/(P)$, with $P$ found by Shoup's deterministic algorithm \cite{Sho90} in time polynomial in $N$ for fixed $q$, and $K=E[Y]/(R)$. If the input must give $\F_{q^m}$ by a single irreducible polynomial $Q$ of degree $m$, one computes $Q$ by the same algorithm and an explicit isomorphism $\psi:K\to\F_q[Z]/(Q)$ by Lenstra's deterministic algorithm \cite{Len91}, and outputs $(\psi(g),\psi(y))$. Since $\psi$ restricts to an automorphism of $\F_q$, it is semilinear over $\F_q$ and preserves rank weights; as $\psi(y-\lambda g)=\psi(y)-\psi(\lambda)\psi(g)$, the two instances have the same answer.
\end{remark}

\begin{table}[!ht]
\centering
\caption{Sizes in the reduction, for a formula with $v$ variables and $c$ clauses, $d=1+v+c$, and fixed $q$.}\label{tab:sizes}
\small
\begin{tabularx}{\textwidth}{@{}l X l@{}}
\toprule
Object & Size & Construction\\
\midrule
exponents $e_i$ & $e_d\le(d-1)^2$ & $\delta_q(i-1)$ (Lemma~\ref{lem:expo})\\
field $E=\F_{q^N}$ & $N=O(qd^2)$, $N\equiv\pm1\pmod 6$ & irreducible polynomial \cite{Sho90}\\
field $K=E(\tau)$ & $m=3N$ & fixed degree-$3$ polynomial\\
$t_k,W,g,y$ & polynomial & linear algebra\\
code & $n=d+r\le N+d$ with $r=\dim W$; $k=1$; radius $n-2$ & $\Gab(g,1)$\\
\bottomrule
\end{tabularx}
\end{table}
\FloatBarrier

\section{Consequences}\label{sec:cons}

\subsection{Complexity consequences}\label{sec:conscx}

We first derive the consequence for Gabidulin codes, then turn to the general problems $\mathsf{RSD}$ and $\mathsf{RMD}$ from Section~\ref{sec:prelim}. Maximum-likelihood decoding of a code consists in finding, for a received word, a closest codeword.

\begin{corollary}\label{cor:ml}
For every fixed prime power $q$, maximum-likelihood decoding of Gabidulin codes is NP-hard under deterministic polynomial reductions, already for codes of dimension $k=1$.
\end{corollary}

\begin{proof}
For $k=1$, a maximum-likelihood decoder applied to $(g,y)$ returns some $\lambda_0$ with $\rk(y-\lambda_0g)=\min_\lambda\rk(y-\lambda g)$. Comparing this rank with $n-2$ decides $\GabDec$, which is NP-complete by Theorem~\ref{thm:main}.
\end{proof}

This holds for the simplest codes of the family, which have dimension one and minimum distance $n$.

We next consider $\mathsf{RSD}$ and $\mathsf{RMD}$.

\begin{corollary}\label{cor:sd}
$\mathsf{RSD}$, the rank syndrome decoding problem for $\F_{q^m}$-linear codes, is NP-complete under deterministic polynomial reductions.
\end{corollary}

\begin{proof}
From an instance $(g,y)$ of $\GabDec$, compute a parity-check matrix $H$ of $\Gab(g,1)$ and $s=Hy^\top$. The solutions $e$ of $He^\top=s$ are exactly the vectors $y-\lambda g$, so a solution of rank at most $n-2$ exists if and only if the instance of $\GabDec$ is a yes-instance.
\end{proof}

\begin{corollary}\label{cor:md}
$\mathsf{RMD}$, the rank minimum distance problem for $\F_{q^m}$-linear codes, is NP-complete under deterministic polynomial reductions, already for codes of dimension $2$.
\end{corollary}

\begin{proof}
From $(g,y)$ consider $C'=\{\mu y+\lambda g:\mu,\lambda\in\F_{q^m}\}$ and $w=n-2$. If $y\in\F_{q^m}g$, which is tested by linear algebra, the instance of $\GabDec$ is a yes-instance and we output any fixed yes-instance; otherwise $C'$ has dimension $2$. Multiplication by a non-zero scalar is an $\F_q$-linear bijection of $\F_{q^m}$, hence preserves rank. A codeword with $\mu=0$ and $\lambda\ne0$ therefore has rank $\rk(g)=n$, while a codeword with $\mu\ne0$ equals $\mu(y-\lambda'g)$ with $\lambda'=-\lambda/\mu$ and has rank $\rk(y-\lambda'g)$. Hence $C'$ has a non-zero codeword of rank at most $n-2$ if and only if the instance of $\GabDec$ is a yes-instance.
\end{proof}

These problems underlie rank-based cryptography. Gaborit and Z\'emor \cite{GZ16} established their hardness through a randomized reduction from Hamming decoding; Corollaries~\ref{cor:sd} and~\ref{cor:md} give NP-completeness under deterministic reductions. Their reduction covers constant rate and small weight, ours is restricted to dimension one or two and to large radius; the two results are complementary. The result can also be seen as the rank-metric counterpart of the Guruswami--Vardy theorem \cite{GV05}.

\subsection{Comparison with the Hamming metric}\label{sec:hamming}

It is instructive to compare with the Hamming metric. There, any linear code of fixed dimension $k$ is decoded at radius $n-k-1$ in polynomial time: a codeword at that distance agrees with the received word on $k+1$ positions, and one may enumerate the $\binom{n}{k+1}$ candidate position sets and solve a linear system for each. For $k=1$ the Reed--Solomon code is the repetition code and the problem is trivial; this is why the Guruswami--Vardy theorem needs a growing dimension. In the rank metric the agreement set is an $\F_q$-subspace rather than a subset, there are $q^{\Theta(kn)}$ candidates, and Theorem~\ref{thm:main} shows that the problem is hard already for $k=1$. The same contrast explains why the Guruswami--Vardy reduction does not transfer: it encodes a subset-sum instance into the freedom of choosing which positions to keep, a freedom which does not exist for subspaces, so that the hard constraints must be written, as here, into the field structure of $\tilde y$.

Several questions remain open and are discussed in Section~\ref{sec:disc}.

\section{Discussion and open questions}\label{sec:disc}

The reduction proves hardness in a specific regime: dimension one, radius $n-2$, and an extension degree $m\approx3n$. Each of these three parameters raises a natural question, namely how far the radius can be lowered, whether the extension degree can be brought down to the values used in cryptography, and whether the dimension can be increased. We discuss them in turn, together with directions that seem interesting to explore.

\paragraph{Below radius $n-2$.} For $k=1$, radius $n-r-1$ asks for an eigenspace of dimension $r+1$ instead of a plane. One might try to use $r+1$ summands in the support, $G=U\oplus\omega_1W\oplus\dots\oplus\omega_rW$, with $\tilde y(a+\sum_i\omega_ib_i)=a^{q^2}+\sum_i\omega_ib_i^q$, where the $\omega_i$ would be powers of $\tau$ whose pairwise products are linearly independent over $E$; a Sidon set of exponents could provide this with an extension of moderate degree. If the coefficientwise reading of the eigenplane equation were to go through as in Section~\ref{sec:gab}, the same pivot problem might give hardness at radius $n-r-1$, at the price of a larger extension degree. Whether this works, down to which radius, and with which trade-off between the radius and $m$, remains to be investigated.

\paragraph{Square codes.} The parameters produced by the reduction are not those most often used in practice: our instances have $m\approx3n$, whereas rank-based cryptosystems frequently work with square or nearly square codes, $m=n$ or $m$ close to $n$, which also give the smallest extension field for a given length. The complexity of decoding in this regime is open. A heuristic count suggests that radius $n-2$ may not be the right place to look for it: a random word has on average about $q^{2(n-2)-m}$ scalars $\lambda$ within that distance, which is large when $m<2(n-2)$. If yes-instances are indeed typical in this regime, which the expectation alone does not prove and which would not exclude worst-case hardness, the natural question for square codes would be at smaller radius.

\paragraph{Higher dimensions.} Decoding $\Gab(g,k)$ at radius $n-k-1$ asks for a $(k+1)$-dimensional space on which $\tilde y$ agrees with a $q$-polynomial of $q$-degree less than $k$. A direction that seems promising is to compose a dimension-one instance with linearized polynomials of $q$-degree one, such as $P(x)=x^q-\delta^{q-1}x$: when $G\subseteq P(\F_{q^m})$, which requires a suitable choice of $\delta$ since $\dim P^{-1}(G)=\dim(G\cap P(\F_{q^m}))+1$, replacing $(G,\tilde y)$ by $(P^{-1}(G),\tilde y\circ P)$ would raise both the dimension of the agreement space and the $q$-degree by one, whereas merely adding coordinates cannot help. The difficulty would be to control the solutions that such a lifting may create without coming from the original instance; if they could be characterized and avoided by a suitable choice of $\delta$, the present core could serve as the base case of an induction on $k$.

\section{Conclusion}

Decoding a Gabidulin code of dimension one at radius $n-2$, the first radius at which the question is not trivial, is NP-complete. Maximum-likelihood decoding of Gabidulin codes is therefore NP-hard already in dimension one, where its Hamming-metric counterpart is easy, and rank syndrome decoding and rank minimum distance for $\F_{q^m}$-linear codes are NP-complete under deterministic reductions, where only a randomized reduction was known. The proof rests on a single idea: once the received word is read as an $\F_q$-linear map, decoding becomes the search for an eigenplane, and two Frobenius twists on the support $U\oplus\tau W$ turn that eigenplane equation into the pivot problem $\Piv$, which encodes 3-SAT.

Our result is a first step towards understanding the complexity of decoding Gabidulin codes beyond the unique decoding radius. It raises natural questions on how far hardness extends, in the dimension of the code, in the decoding radius, and towards the square parameters $m\approx n$ of cryptography, which will require further investigation.

\subsubsection*{Acknowledgment.} AI tools (Claude, Anthropic, and ChatGPT, OpenAI) were used during the preparation of this manuscript for language editing, \LaTeX{} assistance, and exploratory discussions. All mathematical statements, proofs, and conclusions were independently verified by the author, who takes full responsibility for the content.

\end{document}